\documentclass{article}
\usepackage{graphicx}
\usepackage{epstopdf}
\usepackage{amsmath, amsthm, latexsym}
\usepackage{amssymb}
\usepackage{color}

\newtheorem{lemma}{Lemma}

\newtheorem{proposition}{Proposition}
\newtheorem{remark}{Remark}
\usepackage{subfigure}
\begin{document}

\noindent \textbf{\large{Recessional velocities and Hubble's law in Schwarzschild-de Sitter space}}\\

\noindent  \textbf{\normalsize David Klein}\footnote{Department of Mathematics, California State University, Northridge, Northridge, CA 91330-8313. Email: david.klein@csun.edu.}
\textbf{\normalsize and Peter Collas}\footnote{Department of Physics and Astronomy, California State University, Northridge, Northridge, CA 91330-8268. Email: peter.collas@csun.edu.}\\

\vspace{4mm} \parbox{11cm}{\noindent{\small We consider a spacetime with empty Schwarzschild-de Sitter exterior and Schwarzschild-de Sitter interior metric for a spherical fluid with constant density.  The fluid interior may be taken to represent a galaxy supercluster, for which the proper distance from the center of the supercluster to the cosmological horizon has the same order of magnitude as the Hubble radius derived from Friedmann-Robertson-Walker cosmologies.  The fluid interior and surrounding vacuum may also be considered as a model of the Local Group of galaxies in the far future. Particle motion is subject both to the attractive gravity exerted by the fluid and the repelling cosmological constant.  Using  global Fermi coordinates for the central observer within the fluid, the Fermi velocity, the astrometric velocity, the kinematic velocity, and the spectroscopic velocity, relative to the central (Fermi) observer, of a radially receding test particle are calculated and compared.  We find that the Fermi relative velocity can exceed the speed of light in this model, but the presence of a positive cosmological constant causes recessional speeds of distant high energy particles to decrease rather than increase.  We derive a version of Hubble's law for this spacetime which might be applicable for the analysis of a receding mass within a great void adjacent to a supercluster, relatively isolated from gravitational sources other than the supercluster.  We also  compare some of our results to related behavior in FRW cosmologies and consider implications to arguments regarding the expansion of space.}\vspace{5mm}\\


 }\\
\vspace{6cm}
\pagebreak

\setlength{\textwidth}{27pc}
\setlength{\textheight}{43pc}

\noindent \textbf{{\normalsize 1. Introduction}}\\

\noindent The line element for Schwarzschild-de Sitter spacetime with constant density interior is given by,

\begin{equation}\label{metric1}
ds^{2}=
\begin{cases}
 -A(r)d\tilde{t}^{\,2}+B(r)dr^{2}+r^{2}d\Omega^{2} &\text{if}\,\, r\leqslant R\\\\
 -\left(1-\frac{2M}{r}-\frac{\Lambda r^{2}}{3}\right)d\tilde{t}^{2}+\left(1-\frac{2M}{r}-\frac{\Lambda r^{2}}{3}\right)^{-1}dr^{2}+r^{2}d\Omega^{2} &\text{if}\,\,r\geqslant R,\\ 
\end{cases}
\end{equation}

\noindent where $M$ is the mass of the spherical fluid, $\Lambda$ is the cosmological constant, $R$ is the radial coordinate for the radius of the fluid, $d\Omega^{2}=d\theta^{2}+\sin^{2}\theta d\phi^{2}$ and,

\begin{align}
\begin{split}\label{IS2}
A(r)=&\left[\frac{(3-R_{0}^{2}\Lambda)}{2}\sqrt{1-\frac{R^{2}}{R_{0}^{2}}}-\frac{(1-R_{0}^{2}\Lambda)}{2}\sqrt{1-\frac{r^{2}}{R_{0}^{2}}}\right]^{2},\\
B(r)=&\left(1-\frac{r^{2}}{R_{0}^{2}}\right)^{-1}.
\end{split}
\end{align}

\noindent Here,

\begin{equation}
R_{0}^{2}=\frac{3R^{3}}{6M+\Lambda R^{3}}.
\end{equation}

\noindent The metric given by Eq.\eqref{metric1} satisfies the  Israel-Darmois junction conditions and the Einstein field equations (see, e.g., \cite{zarro}) for positive, negative, and zero values of $\Lambda$, but we assume here that $\Lambda \geqslant0$.  We also assume that $A(0)>0$, and that $M,R,\Lambda$ satisfy the generalized Buchdahl inequalities given in \cite{ zarro, Boehmer, boehmer, Hiscock} and the references therein, and later, for given values of $M$ and $R$, we will assume an upper bound on $\Lambda$ (see Eq.\eqref{bounds} below.\footnote{A short calculation shows that if $A(0)>0$ and Eq.\eqref{bounds} holds, then $A(r)>0$ for all $r\in[0, R]$ so that the metric is well defined in the interior region.})\\

\noindent The exterior Schwarzschild-de Sitter metric was used in \cite{orbit, orbit2} to study effects of a positive cosmological constant on the dynamics of the solar system, and some earlier related approaches are summarized in \cite{prd}.  In the present paper, we analyze velocities and accelerations of radially receding distant test particles, relative to the observer at the center of the fluid.\\

\noindent Care is required for the study of relative velocities of non local objects in curved spacetime.  General relativity restricts speeds of test particles to be less than the speed of light, $c=1$, relative to an observer at the exact spacetime point of the test particle.  However, general relativity provides no \textit{a priori} definition of relative velocity, and hence no upper bounds of speeds, for test particles and observers at different spacetime points. Distant particles may have  superluminal \footnote{Here and throughout we define the velocity $v$ of a test particle relative to an observer at a different spacetime point to be \emph{superluminal} if the norm $||v||>1$.} or only sub light speeds, depending on the coordinate system used for the calculations, and on the definition of relative velocity used.  To avoid such ambiguities, we employ the coordinate independent, purely geometric definitions of Fermi, astrometric, kinematic, and spectroscopic relative velocities given in \cite{bolos}, and defined here briefly for the special case of radially receding particles in Schwarzschild-de Sitter space.\\

\noindent The four inequivalent definitions of relative velocity each have physical justifications so as to be regarded as velocities (c.f. \cite{bolos, bolos2}). They  depend on two different notions of simultaneity: ``light cone simultaneity'' and simultaneity as defined by Fermi coordinates of the central observer.  The Fermi and kinematic relative velocities can be described in terms of the latter, according to which events are simultaneous if they lie on the same space slice determined by Fermi coordinates. For a radially receding test particle in this model, the kinematic relative velocity is found by first parallel transporting the four velocity $U$ of the test particle along a radial spacelike geodesic (lying on a Fermi space slice) to a four velocity $U'$ in the tangent space of the central observer, whose four velocity is $u$.  The kinematic relative velocity $v_{\text{kin}}$ is then the unique vector orthogonal to $u$, in the tangent space of the observer, satisfying $U'=\gamma(u+v_{\text{kin}})$ for some scalar $\gamma$ (which is also uniquely determined).  The Fermi relative velocity, $v_{\text{Fermi}}$, under the circumstances considered here, is the rate of change of proper distance of the test particle away from the Fermi observer, with respect to proper time of the observer.\\

\noindent The spectroscopic (or barycentric) and astrometric relative velocities can be derived from spectroscopic and astronomical observations.  Mathematically, both rely on the notion of light cone simultaneity, according to which two events are simultaneous if they both lie on the same past light cone of the central observer. The spectroscopic relative velocity $v_{\text{spec}}$ is calculated analogously to $v_{\text{kin}}$, described in the preceding paragraph, except that the four velocity $U$ of the test particle is parallel transported to the tangent space of the observer along a null geodesic lying on the past light cone of the observer, instead of along the Fermi space slice.  The astrometric relative velocity, $v_{\text{ast}}$, is calculated analogously to $v_{\text{Fermi}}$, as the rate of change of the \textit{observed} proper distance (through light signals at the time of observation) with respect to the proper time of the observer, as may be done via parallax measurements. The observer uses current time measurements together with proper distances of the test particle at the time of emission of light signals, or affine distance.  Details and elaboration may be found in \cite{bolos, bolos2}.\\  

\noindent Analysis of the Fermi relative velocity in Schwarzschild-de Sitter space allows comparisons with the behavior of receding test particles in Friedmann-Robertson-Walker (FRW) cosmologies, where Fermi velocity is (implicitly) used (see, e.g., \cite{confusion, gron}).  We show that the Fermi relative velocity of receding test particles can exceed the speed of light, but together with the astrometric velocity, decreases to zero at the cosmological horizon.  By contrast, the spectroscopic and kinematic relative velocities, which by their definitions cannot exceed the speed of light, reach the speed of light asymptotically at the cosmological horizon.  This property (together with others) of the kinematic velocity makes it a natural choice for the formulation of a version of Hubble's law in this spacetime, a topic developed below.  All relative velocities are calculated with respect to the static observer at $r=0$, who follows a timelike geodesic.\\ 

\noindent In Sec. 2 we express the metric of Eq.\eqref{metric1} using a polar version of Fermi coordinates for the $r=0$ observer.  These Fermi coordinates are global and are convenient for subsequent calculations.  We show that superluminal Fermi relative speeds occur along portions of timelike geodesics at sufficiently high energies and at  large proper distances away from the Fermi observer at $r=0$, even in the Schwarzschild case where $\Lambda=0$. Bounds on the maximum relative Fermi velocities for positive and for zero cosmological constant are also given. We identify a spherical region with radial coordinate $r_{0}$ (at any fixed time) within which test particles initially at rest (at $r < r_{0}$) fall toward the central observer at $r=0$, and outside of which (at $r > r_{0}$) they are accelerated (in Fermi coordinates) in the opposite direction on account of the cosmological constant.  We define the energy, $E_{0}$, of a unit mass test particle at rest at in the spherical region with $r=r_{0}$ to be the \textit{critical energy} of the spacetime; it plays a role in formulating a Hubble's law for Schwarzschild-de Sitter space in Sec. 7.\\ 

\noindent In contrast to the behavior of low energy particles, we also show in Sec. 2 that test particles with high enough energies, following radial geodesics receding from the fluid center at $r=0$, exhibit somewhat counterintuitive behavior.  For such a particle the outgoing Fermi velocity increases in the region $r<r_{0}$ and decreases in the region $r>r_{0}$.  That is, at sufficiently a high energy, the particle, in a certain sense, is ``pushed away'' from the central fluid in the region of space where gravity dominates, and is ``pushed back'' toward the central fluid in the region of space where lower energy particles accelerate \textit{away} from the central fluid due to the influence of the cosmological constant.  A comparison with analogous behavior in FRW cosmologies, identified for example in \cite{gron}, is considered in the concluding section.\\  

\noindent  Sects. 3-5 give formulas for corresponding kinematic, spectroscopic, and astrometric relative velocities of radially receding test particles according to the geometric definitions of \cite{bolos}.  Sec. 6 exhibits functional relationships of the relative velocities, employed in the following section.\\  

\noindent Sec. 7 is devoted to the development of a version of Hubble's law for Schwarzschild-de Sitter space (with strictly positive cosmological constant).  For this purpose, test particles with critical energy $E_{0}$ provide the natural context since in that case the motion of distant particles is due solely to the influence of the cosmological constant. Particles with higher energies may be regarded as having ``peculiar velocities,'' in analogy with FRW models. We derive a linear approximation of the $v_{\text{kin}}$ as a function of proper distance to identify a Hubble's constant in this context.  We then express the redshift of a light signal from a receding particle, relative to the redshift of a static particle at radial coordinate $r_{0}$, in terms of the observed, or affine distance, of the emitting test particle.\\

\noindent In Sec. 8, we consider the spherical fluid as a model for a larger structure, such as a galaxy supercluster. To that end we include numerical results for which the mass of the fluid is $M=10^{3}\,ly$ ($\approx 6\times 10^{15} M_{\odot}$);  $R = 10^{7}\,ly$ ($\approx 3\,Mpc$); and $\Lambda =3\times10^{-20}\; ly^{-2}$.  These choices of parameters are of the same order of magnitude calculated to hold for some galaxy superclusters \cite{Wray, Boerner}.  Moreover, with these parameters, the proper distance, in our model, from the Fermi observer at the center of the fluid to the cosmological horizon is of order $10^{10}$ light years, the same order of magnitude as estimates for the present Hubble length. Included is a discussion of the use of measurements to determine relative velocities and the basic parameters of this model. We also discuss Schwarzschild-de Sitter space as a model of the Local Group of galaxies in the far future.\\

\noindent Concluding remarks and a comparison with recessional velocities in FRW cosmologies, together the implication of our results on the question of the expansion of space, are given in Sec. 9.\\

\noindent \textbf{ 2. Global Fermi coordinates and Fermi relative velocity}\\

\noindent Let $\rho = \rho(r)$ be the proper distance, according to Eq.\eqref{metric1}, from the center of the fluid at $r=0$ to a point with radial coordinate $r$, i.e.,

\begin{equation}\label{proper}
\rho(r)=
\begin{cases}
R_{0}\sin^{-1} (r/R_{0}) &\text{if}\,\, r\leqslant R\\\\
 \int_{R}^{r} \frac{d\tilde{r}}{\sqrt{1-\frac{2M}{\tilde{r}}-\frac{\Lambda \tilde{r}^{2}}{3}}} +R_{0}\sin^{-1} (R/R_{0})&\text{if}\,\,r\geqslant R\,.\\ 
\end{cases}
\end{equation}\\

\noindent In Eq.\eqref{metric1}, we make the change of variable, $t = \sqrt{A(0)}\,\tilde{t}$ and $\rho = \rho(r)$, with the angular coordinates left unchanged.  Denoting the inverse function of $\rho(r)$ by $r(\rho)$, the result is,

\begin{equation}\label{metric2}
ds^{2}=
\begin{cases}
 -\frac{A(r(\rho))}{A(0)}dt^{2}+d\rho^{2}+R^{2}_{0}\sin^{2}(\rho/R_{0})d\Omega^{2} &\text{if}\,\, \rho\leqslant R_{0}\sin^{-1} (R/R_{0})\\\\
 -\left(1-\frac{2M}{r(\rho)}-\frac{\Lambda r(\rho)^{2}}{3}\right)\frac{dt^{2}}{A(0)}+d\rho^{2}+r(\rho)^{2}d\Omega^{2} &\text{if}\,\,\rho\geqslant R_{0}\sin^{-1} (R/R_{0}).\\
\end{cases}
\end{equation}\\

\noindent Note that $\rho(r)$ and all of the metric coefficients in Eq.\eqref{metric2}, in contrast to Eq.\eqref{metric1}, are continuously differentiable, including at the junction, $r(\rho)=R$.  Following standard notation and for later reference, we identify $g_{tt}=g_{tt}(\rho)$ as the metric coefficient of $dt^{2}$ in Eq.\eqref{metric2}, a function of $\rho$ alone, i.e.,

\begin{equation}\label{gtt}
g_{tt}(\rho)=
\begin{cases}
 -A(r(\rho))/A(0) &\text{if}\,\, \rho\leqslant   R_{0}\sin^{-1} (R/R_{0})\\\\
 -\left(1-\frac{2M}{r(\rho)}-\frac{\Lambda r(\rho)^{2}}{3}\right)/A(0) &\text{if}\,\,\rho\geqslant  R_{0}\sin^{-1} (R/R_{0}).\\
\end{cases}
\end{equation}

\noindent It is straightforward to show that the radial spacelike geodesics, orthogonal to the static observer's worldline at $\rho=0$, are of the form,

\begin{equation}\label{geodesic2}
Y(\rho)=(t_{0}, \rho, \theta_{0},\phi_{0}),
\end{equation}

\noindent for any fixed values of $t_{0}, \theta_{0},\phi_{0}$.  With the further change of spatial coordinates, $x^{1}=\rho \sin\theta \cos\phi$, $x^{2}=\rho \sin\theta \sin\phi$, $x^{3}=\rho \cos\theta$, the metric of Eq.\eqref{metric2} is expressed in Fermi coordinates for the static observer at the center of the fluid.  This was proved in \cite{KC3} for the interior part of the metric, and it holds for the metric on the larger spacetime (with the vacuum exterior) considered here.  One may verify that with the above change of variables, the spacelike path below is geodesic and orthogonal to the timelike path of the $\rho=0$ static observer:

\begin{equation}
\begin{split}\label{geodesic3}
Y(\rho)=&(t_{0}, \rho \sin\theta_{0} \cos\phi_{0}, \rho \sin\theta_{0} \sin\phi_{0}, \rho \cos\theta_{0})\\
\equiv&(t_{0}, a^{1}\rho, a^{2}\rho,a^{3}\rho), 
\end{split}
\end{equation}

\noindent for any $t_{0}, \theta_{0}, \phi_{0}$.  The requirement that orthogonal spacelike geodesics have the form of Eq.\eqref{geodesic3} characterizes Fermi coordinates (for background, see \cite{KC1, KC3}). Eq.\eqref{metric2} may thus be regarded as the polar form of the metric in Fermi coordinates.\\

\begin{remark}\label{anti} Replacing $\sin a\rho$ in Eq.\eqref{metric2} by  $\sinh a\rho$ results in the metric for anti-de Sitter space with imbedded constant density fluid expressed in (polar) Fermi coordinates.
\end{remark}

\noindent The following fact, expressed in the form of a lemma, will aid in the physical interpretation of results that follow.

\begin{lemma}\label{A(0)} If $\Lambda\geqslant0$, then $A(0)<1$.
\end{lemma}

\begin{proof} Observe that,

\begin{equation}\label{inequality}
0<\frac{(3-R_{0}^{2}\Lambda)}{2}\sqrt{1-\frac{R^{2}}{R_{0}^{2}}}<\frac{(3-R_{0}^{2}\Lambda)}{2},
\end{equation}
 
\noindent where the first inequality follows from $M>0$.  Subtracting $(1-R_{0}^{2}\Lambda)/2$ yields,

\begin{equation}
-\frac{1}{2}\leqslant-\frac{(1-R_{0}^{2}\Lambda)}{2}<\frac{(3-R_{0}^{2}\Lambda)}{2}\sqrt{1-\frac{R^{2}}{R_{0}^{2}}}-\frac{(1-R_{0}^{2}\Lambda)}{2}<1,
\end{equation}

\noindent from which the result follows. 
\end{proof}

\noindent From Eq.\eqref{metric2}, the Lagrangian for a radial, timelike geodesic is,

\begin{equation}\label{Lagrangian}
L=\frac{g_{tt}\dot{t}^{2}}{2}+\frac{\dot{\rho}^{2}}{2}=-\frac{1}{2},
\end{equation}

\noindent where the overdot signifies differentiation with respect to the proper time $\tau$ along the  geodesic. Since $\partial/\partial t$ is a Killing vector, the energy $E=-p_{t}$ of a unit test particle is invariant along the geodesic, and is given by,

\begin{equation}\label{energy}
p_{t}=g_{tt}\dot{t}=-E.
\end{equation}

\noindent It follows directly from Eqs. \eqref{Lagrangian} and \eqref{energy} that,

\begin{equation}\label{velocity0}
\|v_{\text{Fermi}}\|^{2}=\left(\frac{d\rho}{dt}\right)^{2}
=\frac{\dot{\rho}^{2}}{\dot{t}^{2}}
=-g_{tt}(\rho)\left[1+ \frac{g_{tt}(\rho)}{E^{2}}\right],
\end{equation}

\noindent where we have used Proposition 3 of \cite{bolos} to identify $|d\rho/dt| = \|v_{\text{Fermi}}\|$, the norm of the (geometrically defined) Fermi velocity.    From Eq.\eqref{velocity0} we see that the energy, $E$, of a radial geodesic, passing through a point at proper distance $\rho$ from the central observer, must satisfy,

\begin{equation}\label{restriction}
-g_{tt}(\rho)\leqslant E^{2}.
\end{equation}

\noindent Restricting Eq.\eqref{velocity0} to the exterior region gives,

\begin{equation}\label{velocity}
\|v_{\text{Fermi}}\|^{2}=\left(\frac{d\rho}{dt}\right)^{2}=\frac{\left(1-\frac{2M}{r(\rho)}-\frac{\Lambda r(\rho)^{2}}{3}\right)}{A(0)}\left[ 1- \frac{\left(1-\frac{2M}{r(\rho)}-\frac{\Lambda r(\rho)^{2}}{3}\right)}{A(0)E^{2}}\right].
\end{equation}

\noindent Differentiating Eq.\eqref{velocity} with respect to $t$ gives,

\begin{equation}\label{acceleration}
\frac{d^{2}\rho}{dt^{2}}=\frac{\left(\frac{M}{r(\rho)^{2}}-\frac{\Lambda r(\rho)}{3}\right)}{A(0)}\left[ 1- \frac{2\left(1-\frac{2M}{r(\rho)}-\frac{\Lambda r(\rho)^{2}}{3}\right)}{A(0)E^{2}}\right]r'(\rho),
\end{equation}

\noindent where $r'(\rho) = \sqrt{1-\frac{2M}{r(\rho)}-\frac{\Lambda r(\rho)^{2}}{3}}$ follows from Eq.\eqref{proper}.  The acceleration according to the Fermi coordinates of the central observer therefore vanishes at up to three values of $\rho$: (\textit{a}) at the cosmological horizon (where  $r'(\rho) =0$); (\textit{b}) if,

\begin{equation}\label{acceleration2}
1-\frac{2M}{r(\rho)}-\frac{\Lambda r(\rho)^{2}}{3}= \frac{A(0)E^{2}}{2},
\end{equation}

\noindent and, assuming $\Lambda>0$, (\textit{c}) at,

\begin{equation}\label{r0}
r(\rho_{0})=r_{0}\equiv \left(\frac{3M}{\Lambda}\right)^{1/3}.
\end{equation}

\noindent Henceforth, we assume, 

\begin{equation}\label{bounds}
 0\leqslant\Lambda < \min\left(\frac{1}{9M^{2}}, \frac{3M}{R^{3}}\right). 
\end{equation}

\noindent Inequality \eqref{bounds} guarantees that $r_{0}>R$ and $1-\frac{2M}{r_{0}}-\frac{\Lambda r_{0}^{2}}{3}>0$  so that $r_{0}$ lies in the exterior vacuum somewhere between the boundary of the fluid and the cosmological horizon of the Fermi observer. This natural condition is fulfilled by our examples.\\

\noindent The number $r_{0}$ is the radial coordinate where gravitational attraction is exactly balanced by repulsion from the cosmological constant.  To elaborate on this point, we define the \textit{critical energy}, $E_{0}$, by,

\begin{equation}\label{rest}
E_{0}^{2}A(0) = 1-\frac{2M}{r_{0}}-\frac{\Lambda r_{0}^{2}}{3}.
\end{equation}

 \noindent It is easily checked that a particle with energy $E_{0}$ at radial coordinate $r_{0}$ has zero Fermi velocity and zero acceleration, and remains at rest.  The gravitational acceleration inward is exactly balanced by the acceleration outward due to the cosmological constant.  A particle initially at rest at a point closer to the the central observer (with initial coordinates satisfying $r(\rho)< r_{0}$) will accelerate toward the central observer, while a particle initially at rest with radial coordinate larger than $r_{0}$ will accelerate away from the central observer, in Fermi coordinates.  We note that in the standard weak field approximation for the Newtonian potential energy function via $1+2V/c^{2} = -g_{tt}$ (where $c$ is the speed of light),
 
\begin{equation}\label{potential}
V(r) = -\frac{GM}{r}-\frac{\Lambda c^{2} r^{2}}{6},
\end{equation}
 
\noindent so that the force $F$ is given by,

\begin{equation}\label{force}
F(r)=- \nabla V(r) = -\frac{GM}{r^{2}}+\frac{\Lambda c^{2} r}{3}.
\end{equation}

\noindent Setting $F(r) = 0$ yields the same expression for $r_{0}$ as in Eq.\eqref{r0}, though in the relativistic case the proper distance from the central observer is $\rho(r_{0})$ as given by Eq.\eqref{proper}.\\

\begin{figure}[!h]
\begin{center}
\resizebox{!}{5.5 cm}{
$$\includegraphics{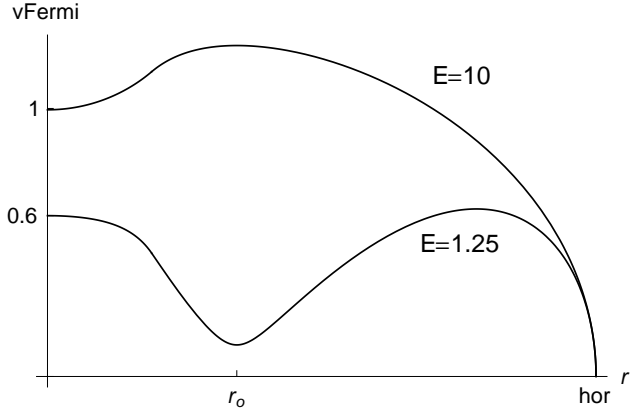}$$}
\caption{Low and high energy $v_{\text{Fermi}}$ for  $M=20,\, R=100,\, \Lambda = 10^{-5}$. Here, $r_{0}\approx181.7$, $E_{0}\approx1.24$.}\label{fig1}
\end{center}
\end{figure}

\noindent A particle with energy $E >E_{0}$ satisfies Eq.\eqref{restriction} in the entire vacuum region of the spacetime.  From Eq.\eqref{acceleration} it follows that if $E_{0}<E <\sqrt{2}\,E_{0}$, the test particle decelerates before it reaches a distance with radial coordinate $r_{0}$, and as soon as it passes that point, it begins to accelerate away from the fluid toward the cosmological horizon. This acceleration toward the horizon continues until the factor in the square brackets on the right side of Eq.\eqref{acceleration} reaches zero, at which point the particle decelerates. However, if $E>\sqrt{2}\,E_{0}$ the opposite occurs: the particle accelerates before it reaches a distance with radial coordinate $r_{0}$, and thereafter decelerates.  In both scenarios, the particle's relative Fermi velocity decreases to zero at the cosmological horizon.  The effect of the cosmological constant is strikingly different in these two cases.  Fig. \ref{fig1} illustrates these general features  for particular (though artificial) choices for the parameters.  Note that the initial velocity of the high energy particle ($E=10$) is slightly below the speed of light. In the case that $\Lambda=0$, it is not difficult to verify that for high energy unit mass particles, with $A(0)E^{2} \geqslant 2$, the outward acceleration given by Eq.\eqref{acceleration} is positive throughout the exterior vacuum. Thus, the negative acceleration of the high energy particle for $r>r_{0}$ in Fig.\ref{fig1} is due to a positive cosmological constant (and is not merely a property of Fermi coordinates).\\

\noindent We conclude this section with some observations in the form of two propositions.\\
 
\begin{proposition}\label{superluminal}  Assume that $\Lambda >0$. As above, let $r_{0}\equiv (3M/\Lambda)^{1/3}$ and let $v_{\text{Fermi}}$ denote the Fermi velocity, relative to the central observer, of a test particle receding radially along a timelike geodesic in the exterior vacuum of Schwarzschild-de Sitter spacetime. Then\\ 
 
\noindent (a) For any energy $E$ of the test particle, $\|v_{\text{Fermi}}\| < E_{0}$ along its geodesic in the exterior vacuum.\\
 
\noindent (b) The maximum value of $\|v_{\text{Fermi}}\|$ as a function of $\rho$ exceeds the speed of light for sufficiently high energy $E$ if and only if $E_{0}>1$, i.e., $A(0)<1-\frac{2M}{r_{0}}-\frac{\Lambda r_{0}^{2}}{3}.$
 \end{proposition}
 
\begin{proof} Part (a) follows from Eq.\eqref{velocity} and the easily verified fact that the function, $1-\frac{2M}{r(\rho)}-\frac{\Lambda r(\rho)^{2}}{3}$, achieves its maximum value at $r(\rho)= r_{0}$.  It then follows from part (a) that $A(0)<1-\frac{2M}{r_{0}}-\frac{\Lambda r_{0}^{2}}{3}$ is a necessary condition for $\|v_{\text{Fermi}}\|$ to exceed the speed of light at some point on the radial geodesic. Sufficiency follows by taking a limit of $\|v_{\text{Fermi}}\|$ evaluated at $r(\rho)=r_{0}$ using Eq.\eqref{velocity}

\begin{equation}\label{limit}
\lim_{E\rightarrow\infty}\|v_{\text{Fermi}}\|=\sqrt{\frac{1-\frac{2M}{r_{0}}-\frac{\Lambda r_{0}^{2}}{3}}{A(0)}}=E_{0}.
\end{equation}

\end{proof}

\begin{proposition}\label{superluminal2} Let $\Lambda =0$ and assume that $M, R$ satisfy the Buchdahl inequality, $M/R<4/9$. A  test particle in the exterior vacuum of Schwarzschild spacetime, receding radially along a timelike geodesic will achieve a Fermi velocity in excess of the speed of light, relative to the Fermi observer at the fluid center,  for all sufficiently high energies and sufficiently large proper distances from the fluid center.  The Fermi relative speed, $\|v_{\text{Fermi}}\|$, is bounded above by $\sqrt{1/A(0)}$.
\end{proposition}

\begin{proof} $\|v_{\text{Fermi}}\|<\sqrt{1/A(0)}$ follows directly from Eq.\eqref{velocity} in the case that $\Lambda=0$.  It also follows from Eq.\eqref{velocity} that,

\begin{equation}\label{velocityschwarz}
\lim_{\rho\rightarrow\infty}\|v_{\text{Fermi}}\|=\sqrt{\frac{1}{A(0)}\left[ 1- \frac{1}{A(0)E^{2}}\right]}>1,
\end{equation}

\noindent for $E$ sufficiently large, since by Lemma \ref{A(0)}, $A(0)$ is necessarily less than $1$.  

\end{proof}

\noindent Thus, even in Schwarzschild spacetime the Fermi relative velocity of a radially receding particle, far from the central observer,  can exceed the speed of light. We discuss the significance of this in the concluding section.

\begin{remark}\label{light}
The speed of a radially receding distant photon, with respect to proper time and proper distance of the central observer in the fluid, i.e., in Fermi coordinates, may be computed by setting the right side of Eq.\eqref{Lagrangian} equal to zero.  The speed of the photon in Fermi coordinates is thus $\sqrt{-g_{tt}}$, which is an upper bound and limiting value for the Fermi speed of a massive particle, given by Eq.\eqref{velocity0}, at the same spacetime position, as must be the case. The maximum possible relative Fermi speed of a distant photon is therefore the critical energy (per unit momentum) $E_{0}$.
\end{remark}

\noindent \textbf{3. Kinematic relative velocity}\\

\noindent  The four-velocity of the central observer at $\rho=0$ is $u=(1,0,0,0)$.  Let $U=U(\rho)$ denote the four-velocity along a timelike radial geodesic of a radially receding test particle at a proper distance $\rho$ from the central observer.  Without loss of generality, we assume $\theta=\phi=0$. It follows from Eqs.\eqref{Lagrangian} and \eqref{energy} that,

\begin{align}\label{vk2b}
U(\rho)&=\left(-\frac{E}{g_{tt}(\rho)}\,,\,\sqrt{-\frac{E^{2}}{g_{tt}(\rho)}-1}\,,\,0\,,\,0\right).
\end{align}

\noindent We assume that  $E >E_{0}$ so that Eq.\eqref{restriction} holds and Eq.\eqref{vk2b} is well-defined throughout the exterior (vacuum) region.\\

\noindent The kinematic relative velocity $v_{\text{kin}}$ of $U$ with respect to the central observer's four-velocity $u$ is given by (see \cite{bolos}),

\begin{equation}
\label{vk1}
v_{\text{kin}}=\frac{1}{-g\left(\tau_{\rho0}U,u\right)}\tau_{\rho0}U-u\,,
\end{equation}

\noindent where $g$ is the bilinear form defined by Eq.\eqref{metric2} and $\tau_{\rho0}U$ is the parallel transport of $U$ from a proper distance $\rho$ to the fluid center, $\rho =0$, along the spacelike radial geodesic  with tangent vector $X=(0,1,0,0)$, connecting these two points.  It follows from its definition that the kinematic speed, i.e., the norm of the kinematic velocity, cannot exceed the speed of light.\\

\noindent Since the affine coefficients $\Gamma^{\rho}_{\;t\rho}=\Gamma^{\rho}_{\;\rho\rho}\equiv0$ for the metric of Eq.\eqref{metric2}, it is easily verified that the parallel transport of the $\rho$-component, $U^{\rho}$, of $U$ is constant along spacelike radial geodesics.  Thus $(\tau_{\rho0}U)^{\rho}=U^{\rho} (\rho)$. Also, it follows from symmetry and Eq.\eqref{vk2b} that the angular components of the parallel transport of $U$ are zero along the radial spacelike geodesic.  At the origin, $\rho=0$, Eq.\eqref{metric2} becomes the Minkowski metric, so \footnote{Spherical polar coordinates are singular at $\rho=0$, but $(\tau_{\rho0}U)^{\rho}$ is meaningful as a limit as $\rho\rightarrow0$.  Alternatively, if standard (Cartesian) Fermi coordinates, via the coordinate transformation identified above Eq.\eqref{geodesic3}, are used, the radial direction may be identified as the  $x,y$ or $z$ axis in the usual Minkowski coordinates at the center of the fluid, in which case $(\tau_{\rho0}U)^{\rho}$ is well-defined.},

\begin{equation}\label{minkowski}
-\left((\tau_{\rho0}U)^{t}\right)^{2}+\left((\tau_{\rho0}U)^{\rho}\right)^{2}=-1
\end{equation}

\noindent Thus,

\begin{equation}
\label{vk4}
(\tau_{\rho0}U)^{t}=\sqrt{1+\left(U^{\rho}(\rho)\right)^{2}}=\frac{E}{\sqrt{-g_{tt}(\rho)}}\,.
\end{equation}

\noindent We then find,
\begin{equation}
\label{vk5}
\tau_{\rho0}U=\left(\frac{E}{\sqrt{-g_{tt}(\rho)}}\,,\,\sqrt{-\frac{E^{2}}{g_{tt}(\rho)}-1}\,,\,0\,,\,0\right)\,,
\end{equation}

\noindent and using Eq.\eqref{vk1}, we find that the kinematic speed as a function of $\rho$ is given by,

\begin{equation}
\label{vk6}
\|v_{\text{kin}}\|=\sqrt{1+\frac{g_{tt}(\rho)}{E^{2}}}\,.
\end{equation}

\noindent \textbf{4. Gravitational Doppler Shift and Spectroscopic Relative Velocity}\\

\noindent The gravitational Doppler shift of a test particle receding from the observer at $\rho=0$ is given by 
\begin{equation}\label{doppler}
\frac{\nu_{R}}{\nu_{E}}=\frac{p_{\mu}(R)u^{\mu}(R)}{p_{\mu}(E)u^{\mu}(E)},
\end{equation}

\noindent where ``$E$'' refers to emitter and ``$R$'' to receiver, so that $\nu_{E}, p_{\mu}, (E),u^{\mu}(E)$ represent respectively the frequency of an emitted photon from the receding test particle, the four-momentum of the emitted photon, and the four-velocity of the receding test particle, with analogous definitions for the remaining terms.  The four-momentum (as a four-vector) of a photon traveling toward the observer at $\rho =0$ is given by,
\begin{equation}\label{photon}
p=\left(\frac{p_{t}}{g_{tt}},\frac{p_{t}}{\sqrt{-g_{tt}}},0,0\right),
\end{equation}

\noindent where the energy, $-p_{t}$, of the photon is constant along the null geodesic. The four-velocity of the test particle is,
\begin{equation}\label{4vel}
u= (\dot{t},\dot{\rho},0,0)= \left(\frac{-E}{g_{tt}},  \sqrt{-1-\frac{E^{2}}{g_{tt}}},  0,0\right)
\end{equation}

\noindent Combining Eqs.\eqref{doppler}, \eqref{photon}, \eqref{4vel}, gives,
\begin{equation}\label{shift}
\frac{\nu_{R}}{\nu_{E}}= \frac{-g_{tt}}{E\left[1+ \sqrt{1+\frac{g_{tt}}{E^{2}}}\right]},
\end{equation}

\noindent where $g_{tt}$ is evaluated at the point of emission of the photon at the location of the receding test particle.  The spectroscopic relative velocity, as defined in \cite{bolos}, may be computed for the case of a particle receding from the origin, directly from Eq.(12) of \cite{bolos} and Eq.\eqref{shift} above.
\begin{equation}
\label{vspec}
\|v_{\text{spec}}\|=\frac{(\nu_{E}/\nu_{R})^{2}-1}{(\nu_{E}/\nu_{R})^{2}+1}\,.
\end{equation}

\noindent \textbf{5. Astrometric Relative Velocity}\\

\noindent A photon with unit energy receding radially in the past-pointing horismos (i.e. backward light cone) of the observer at the center of the fluid, with spacetime path $\sigma(t)=(t,0,0,0)$, has four-momentum (as a four-vector) given by, 

\begin{equation}\label{photon2}
p=(\dot{t}, \dot{\rho},0,0)=\left(\frac{1}{g_{tt}},\frac{1}{\sqrt{-g_{tt}}},0,0\right),
\end{equation}

\noindent where the overdot represents differentiation with respect to an affine parameter $\lambda$.   Let $N(\lambda)$ be the null, past-pointing, geodesic with $N(0)=\sigma(t)$ and tangent vector given by Eq.\eqref{photon2} so that $dN/d\lambda (0)=p(0)=(-1,1,0,0)$ and $N(\lambda)=\exp_{\sigma(t)}(\lambda\, p(0))$.\\

\noindent The (past-pointing) photon departing from the observer $\sigma(t)$ at time $t$ will intersect the worldline of the receding test particle determined by Eqs.\eqref{energy} and \eqref{velocity} at a spacetime point $(t^{*}, \rho^{*},0,0)$, where $t^{*}$ is a unique time in the past of the observer $\sigma(t)$, and $\rho^{*}\equiv \rho(t^{*})$.\\

\noindent The affine distance $d^{\text{aff}}$ from the observer $\sigma(t)$ to the spacetime point $(t^{*}, \rho^{*},0,0)$ is defined as the norm of the projection of $\exp_{\sigma(t)}^{-1}[(t^{*},\rho^{*},0,0)]$ onto the orthogonal complement $\sigma'(t)^{\bot}$ of $\sigma'(t)$.  The astrometric speed for the radially recessing test particle is $d(d^{\text{aff}})/dt$.  To compute this we use the easily verified fact that  $N(d^{\text{aff}}) = (t^{*},\rho^{*},0,0)$ (see Eq.(16) and Propositions 6 and 7 of \cite{bolos}).\\

\noindent To see that $d^{\text{aff}} = \rho^{*}$, let $t(\rho)$ be the inverse function of $\rho(t)$ and observe that,
\begin{equation}\label{t}
t(\rho^{*})=t^{*} =t + \int^{d^{\text{aff}}}_{0}\frac{dt}{d\lambda}d\lambda = t(d^{\text{aff}}) 
\end{equation}

\noindent Thus, since $t(\rho)$ is one-to-one, $d^{\text{aff}} = \rho^{*}$. From Eq.\eqref{photon2} it follows that,
\begin{equation}\label{photon3}
\frac{dt}{d\rho}= \frac{-1}{\sqrt{-g_{tt}}}<0.
\end{equation}

\noindent Now, using Eq.\eqref{photon3} we find,

\begin{equation}\label{t*}
t=t^{*}+\int_{\rho^{*}}^{0}\frac{dt}{d\rho}d\rho=t^{*}+\int^{\rho(t^{*})}_{0}\frac{1}{\sqrt{-g_{tt}}}d\rho,
\end{equation}

\noindent which determines $t$ as a function of $t^{*}$ and therefore determines the inverse function, $t^{*}(t)$ as well. Using the chain rule and Eq.\eqref{t*}, it follows that the astrometric relative velocity $v_{\text{ast}}$ is given by,

\begin{equation}\label{ast}
\|v_{\text{ast}}\|=\frac{d\rho(t^{*}(t))}{dt}= \rho'(t^{*})\frac{dt^{*}}{dt}=\frac{\rho'(t^{*})}{1+\frac{\rho'(t^{*})}{\sqrt{-g_{tt}(\rho^{*})}}},
\end{equation}

\noindent with motion in the radial direction. The astrometric relative velocity may be computed for a given value of $t$ by first using Eq.\eqref{t*} to determine $t^{*}$ numerically and then Eq.\eqref{ast}. Since $g_{tt} \rightarrow 0$ at the cosmological horizon, the astrometric relative velocity is asymptotically zero for high energy test particles. \\

\begin{remark}\label{astro} For a test particle approaching, rather than receding from, the central fluid radially, the right side of Eq.\eqref{ast} is changed by a factor of $-1$.  In that case, the astrometric speed can exceed $1$, as in the case for Minkowski space, illustrated in \cite{bolos}.
\end{remark}

\noindent \textbf{{\normalsize 6. Functional relationships between the relative velocities}}\\

\noindent In this section we identify some functional relationships between the four relative velocities. The Fermi and kinematic relative velocities are closely related. Observe that by Eqs.\eqref{velocity0} and \eqref{vk6}, 

\begin{equation}\label{compare5}
\|v_{\text{Fermi}}\|=\sqrt{-g_{tt}}\, \|v_{\text{kin}}\|.
\end{equation}

\noindent The relationship between astrometric and Fermi velocities follows directly from Eq.\eqref{ast},

\begin{equation}\label{compare4}
\|v_{\text{ast}}(t)\|=\frac{\|v_{\text{Fermi}}(t^{*})\|}{1+\frac{\|v_{\text{Fermi}}(t^{*})\|}{\sqrt{-g_{tt}(\rho(t^{*}))}}},
\end{equation}

\noindent where the left side of Eq.\eqref{compare4} is evaluated on the worldline of the central Fermi observer at $\sigma(t)=(t,0,0,0)$, and on the right side, the Fermi velocity is evaluated at the spacetime point $(t^{*}, \rho(t^{*}),0,0)$ in the past light cone of the Fermi observer.  The functional relationship between $t$ and $t^{*}$ is given by Eq.\eqref{t*}.  Combining Eqs.\eqref{compare5} and \eqref{compare4} yields,

\begin{equation}\label{compare6}
\|v_{\text{ast}}(t)\|=\frac{\sqrt{-g_{tt}(\rho(t^{*}))}\,\|v_{\text{kin}}(t^{*})\|}{1+\|v_{\text{kin}}(t^{*})\|},
\end{equation}

\noindent so that a present measurement of the astrometric velocity is determined by the kinematic velocity at a spacetime point in the past lightcone.\\

\noindent From \eqref{vk6}, and \eqref{shift}, we also have,

\begin{equation}\label{compare1}
\|v_{\text{kin}}\| =- \left(1+\frac{g_{tt}}{E}\frac{\nu_{E}}{\nu_{R}}\right),
\end{equation}

\noindent where $g_{tt}$ is evaluated at the location of the test particle and emission of a photon. As in the preceding case, some care is required in the interpretation of the terms in Eq.\eqref{compare1} as functions of time (as opposed to radial distance $\rho$).  This is because $v_{\text{kin}}$ is the relative velocity at the time of emission of the photon, which is received and whose frequency, $\nu_{R}$, is measured by the central observer only at a later time.\\  

\noindent From Eq.\eqref{vspec}, it follows that,

\begin{equation}\label{compare2}
\left(\frac{\nu_{E}}{\nu_{R}}\right)^{2}=\frac{1+\|v_{\text{spec}}\|}{1-\|v_{\text{spec}}\|}.
\end{equation}

\noindent Combining this with Eq. \eqref{compare1} yields an expression for $\|v_{\text{kin}}\|$ in terms of $\|v_{\text{spec}}\|$,

\begin{equation}\label{compare3}
\|v_{\text{kin}}(t^{*})\|=-\frac{g_{tt}}{E}\sqrt{\frac{1+\|v_{\text{spec}}(t)\|}{1-\|v_{\text{spec}}(t)\|}}-1,
\end{equation}

\noindent where as above, the time of evaluation of the right side is in the future of the time of evaluation of the left side.\\

\noindent Observe now that by combining Eqs.\eqref{compare3} and \eqref{compare6}, the astrometric velocity may be expressed directly in terms of the spectroscopic velocity as, 

\begin{equation}\label{compare7}
\|v_{\text{ast}}(t)\|=\sqrt{-g_{tt}(\rho^{*})}\,-\frac{E}{\sqrt{-g_{tt}(\rho^{*})}}\sqrt{\frac{1-\|v_{\text{spec}}(t)\|}{1+\|v_{\text{spec}}(t)\|}}
\end{equation}

\noindent where $\rho^{*}=\rho(t^{*})$ is the affine distance as in Sec. 5, i.e., the proper distance observed at the time of sighting.  Eq.\eqref{compare7} together with Eq.\eqref{vspec} provides a way to compare, in principle, spectroscopic and parallax measurements for radially receding particles.\\

\noindent \textbf{{\normalsize 7. Hubble's Law}}\\

\noindent In this section, we derive two versions of Hubble's law for Schwarzschild-de Sitter space, with $\Lambda> 0$, using linear approximation of the dependence of $\|v_{\text{kin}}\|$ on proper distance. For the energy of the receding test particle, we take $E=E_{0}$, given by Eq.\eqref{rest}. This is physically natural because $E_{0}$ is the minimum energy of a test particle that does not fall back into the central fluid. Recall from Sec. 1 that a particle with critical energy $E_{0}$ remains at rest at a point with radial coordinate $r_{0}$, but starting at any position with radial coordinate $r> r_{0}$ it will recede from the central observer.  The radial velocity of such a test particle is due solely to the cosmological constant, and not what might be described as an initial ``peculiar'' velocity.\\

\noindent From Eqs. \eqref{rest} and \eqref{vk6} with $E=E_{0}$,

\begin{equation}\label{hubble1}
\|v_{\text{kin}}(\rho)\|^{2}=1-\frac{1-\frac{2M}{r(\rho)}-\frac{\Lambda r(\rho)^{2}}{3}}{1-\frac{2M}{r_{0}}-\frac{\Lambda r_{0}^{2}}{3}}.
\end{equation}

\noindent Expanding Eq.\eqref{hubble1} in a Taylor series centered at $\rho_{0}=\rho(r_{0})$ (so that $r(\rho_{0})=r_{0}$) gives,

\begin{equation}\label{hubble2}
\|v_{\text{kin}}(\rho)\|^{2}\approx\|v_{\text{kin}}(\rho_{0})\|^{2}+\left(\frac{2M}{r_{0}^{3}}+\frac{\Lambda}{3}\right)(\rho-\rho_{0})^{2}.
\end{equation}

\noindent By Eq.\eqref{r0} and the fact that $\|v_{\text{kin}}(\rho_{0})\|=0$ when $E=E_{0}$, we have,

\begin{equation}\label{hubble3}
\|v_{\text{kin}}(\rho)\|\approx\sqrt{\Lambda}(\rho-\rho_{0}), 
\end{equation}

\noindent valid for a distance $\rho$ close to $\rho_{0}$, the balance point between gravitational attraction and repulsion due to the cosmological constant (for the qualitative behavior of $\|v_{\text{kin}}\|$ as a function of distance, see Fig. 2A below).  We may thus define a ``Hubble constant'' $H$ for Schwarzschild-de Sitter spacetime by,

\begin{equation}\label{hubble4}
H=\sqrt{\Lambda}.
\end{equation}

\noindent For example, if $\Lambda=10^{-20}\,ly^{-2}$, according to Eq.\eqref{hubble4}, $H=10^{-10}\,ly^{-1}$ which is the same order of magnitude as current measurements of the  Hubble constant ($H_{0}\approx7.2\times10^{-11}\,ly^{-1}\approx 70\, km/s/Mpc$).\\

\noindent For large distances, on the order of magnitude of the distance to the cosmological horizon, which roughly coincides with the Hubble radius when parameters for a galaxy supercluster are taken (see the following section), a linear approximation more accurate than Eq.\eqref{hubble3} is,

\begin{equation}\label{hubble5}
\|v_{\text{kin}}(\rho)\|\approx \frac{1}{\rho_{\text{horizon}}}(\rho-\rho_{0})\approx  \frac{1}{\rho_{\text{horizon}}}\rho,
\end{equation}

\noindent where $\rho_{\text{horizon}}$ is the proper distance from the central observer to the cosmological horizon.  This choice of linear approximation forces $\|v_{\text{kin}}(\rho)\|\rightarrow1$ as $\rho\rightarrow\rho_{\text{horizon}}$.  For the model of a galaxy supercluster considered in the next section, $\rho_{\text{horizon}}=1.57\times10^{10}\,ly$.\\

\noindent To obtain a formula for the redshift of a photon in terms of the affine distance of the emitter, $\rho^{*}$, we first combine Eqs.\eqref{hubble3} and \eqref{compare1} to get,

\begin{equation}\label{hubble6}
- 1-\frac{g_{tt}(\rho^{*})}{E_{0}}\frac{\nu_{E}}{\nu_{R}}\approx\sqrt{\Lambda}(\rho^{*}-\rho_{0}),
\end{equation}

\noindent with the same notation as in the previous section.  Using the fact that the Taylor expansion of $-g_{tt}(\rho^{*})/E_{0}$ about $\rho_{0}$ is given by (see Eqs.\eqref{r0} and \eqref{rest}),

\begin{equation}
\frac{-g_{tt}(\rho^{*})}{E_{0}}= E_{0}+0(\rho^{*}-\rho_{0})+\mathcal{O}(2)= E_{0}+\mathcal{O}(2),
\end{equation}

\noindent and rearranging terms results in,

\begin{equation}\label{hubble7}
\frac{\nu_{E}}{\nu_{R}}\approx\frac{\sqrt{\Lambda}}{\,\,E_{0}}(\rho^{*}-\rho_{0})+\frac{1}{E_{0}}.
\end{equation}

\noindent A physical interpretation of the last term on the right side of Eq.\eqref{hubble7} may be given.  A short calculation using Eqs.\eqref{doppler} and \eqref{photon} shows that,

\begin{equation}\label{hubble8}
\frac{1}{E_{0}}=-\frac{E_{0}}{g_{tt}(\rho_{0})}=\left(\frac{\nu_{E}}{\nu_{R}}\right)_{0},
\end{equation}

\noindent where $z_{0}\equiv\left(\frac{\nu_{E}}{\nu_{R}}\right)_{0}-1$ is the redshift of a photon measured by the central Fermi observer at $\rho=0$ and emitted by a stationary observer with energy $E_{0}$ at a fixed point in space at proper distance $\rho_{0}$ from the central observer (i.e., at a point with radial coordinate $r_{0}$). Thus, denoting the redshift factor, as is customary, by $z= \nu_{E}/\nu_{R}-1$ gives,

\begin{equation}\label{hubble9}
z-z_{0}\approx \frac{\sqrt{\Lambda}}{\,\,E_{0}}\,(\rho^{*}-\rho_{0}).
\end{equation}

\noindent For the parameters used in the following section to model a galaxy supercluster, $E_{0}= 1.00012$ (which by Remark \ref{light} is the maximum Fermi relative speed of a photon) so that the ``Hubble constant'' of Eq.\eqref{hubble9} or Eq.\eqref{hubble7} has the same order of magnitude as in Eq.\eqref{hubble3}.\\

\noindent \textbf{{\normalsize 8. Particles receding from a galaxy supercluster}}\\

\begin{figure}[!h]
\centering
\mbox{\subfigure{\includegraphics[width=6cm]{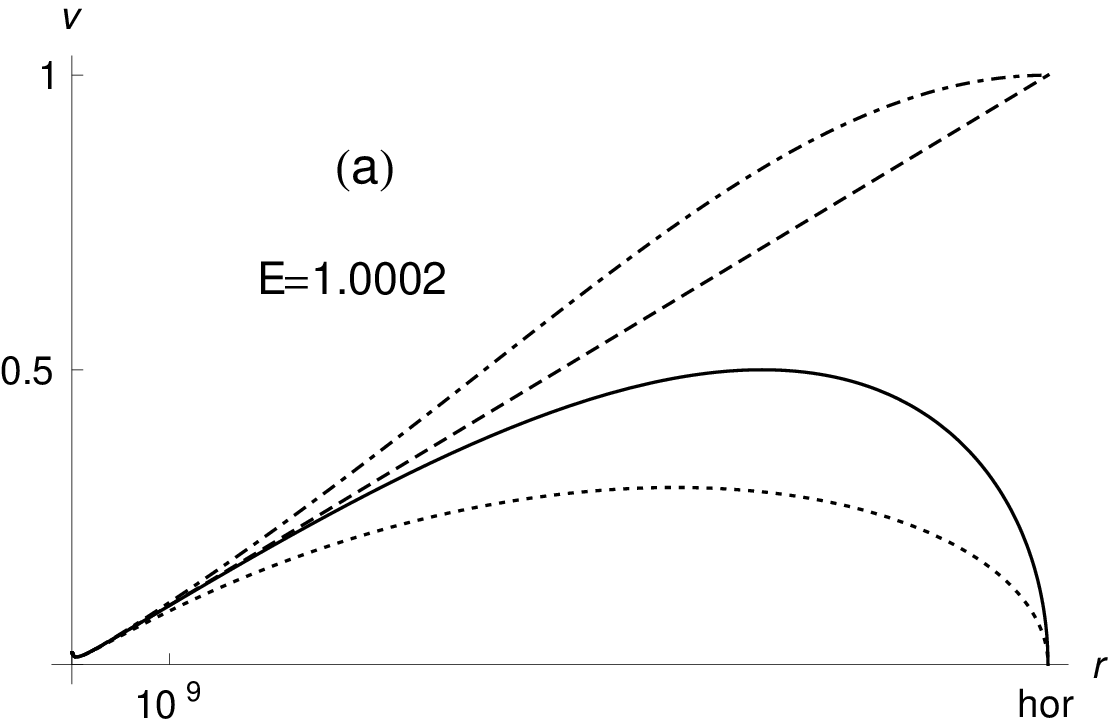}
\subfigure{\includegraphics[width=6cm]{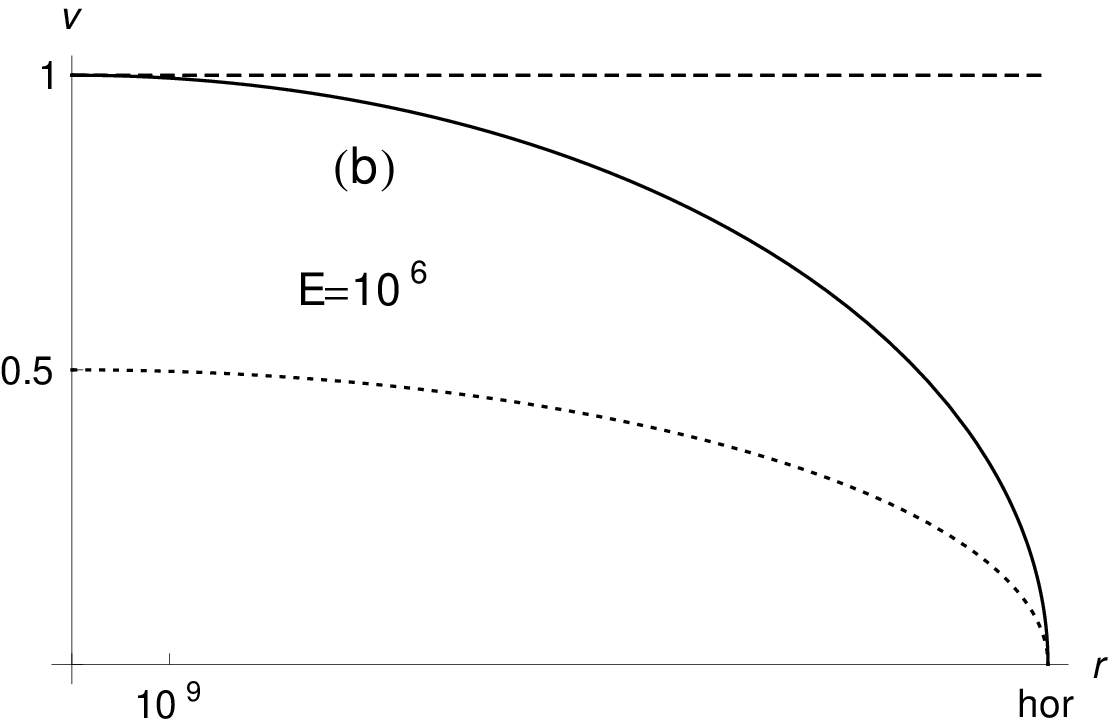} }}}
\caption{Low and high energy behavior of the velocities. $v_{\text{ast}}$ (dotted), $v_{\text{Fermi}}$ (solid), $v_{\text{kin}}$ (dashed), $v_{\text{spec}}$ (dot-dashed). At high $E$, $v_{\text{kin}}\approx v_{\text{spec}}$ (dashed).  $E_{0}\approx 1.00012$.}
\label{fig2}
\end{figure}

\noindent In this section we compare the Fermi, astrometric, kinematic, and the spectroscopic velocities of a radially receding test particle, relative to the observer at the center of the fluid, for specific values of the parameters of Eq.\eqref{metric2}.  We let $M=10^{3}\,ly$, $R = 10^{7}\,ly$, and $\Lambda =3\times 10^{-20}\,ly^{-2}$.  As noted in the introduction, these choices for the parameters are of the same order of magnitude as those for some galaxy superclusters \cite{Wray, Boerner}.  It is then readily deduced that $r_{0}\approx 4.6\times 10^{7}$ ly, and $E_{0}= 1.00012$ (c.f. Eq.\eqref{rest}).  The radial coordinate of the horizon is obtained by solving $g_{tt}(r)=0$ for $r$ and yields $r\equiv hor \approx 10^{10}\,ly$. It then follows from Eq.\eqref{proper} that the proper distance, in this model, from the observer at the center of the fluid to the cosmological horizon, is roughly $1.57\times10^{10}\,ly$, the same order of magnitude as estimates for the present Hubble length.\\

\noindent Figs. \ref{fig2}(a) and \ref{fig2}(b) give graphical comparisons of the Fermi, astrometric, kinematic velocity, and spectroscopic velocities for a receding test particle relative to the central observer at $\rho=0$, at low and high energies. Fig. \ref{fig2}(a) shows how the kinematic and spectroscopic relative velocities reach the speed of light at the horizon, while the Fermi and astrometric relative velocities decrease to zero at the horizon.  Notice in particular the nearly linear behavior of $\|v_{\text{kin}}\|$ with respect to $r$ (which is the case also with respect to $\rho$), consistent with the ``Hubble law'' given in the previous section through linear approximation.\\

\noindent In Fig. \ref{fig2}(b), the kinematic and spectroscopic relative velocities nearly coincide and are nearly equal to the speed of light. The Fermi velocity of a particle of unit rest mass with sufficiently high energy can slightly exceed the speed of light, but by no more than $E_{0}=1.00012$, as follows from Proposition \ref{superluminal}. The qualitative behavior is the same as in Fig. \ref{fig1}.\\

\noindent The use of Eq.\eqref{metric1} to model a galaxy supercluster and its surrounding vacuum has evident shortcomings.  The absence of other gravitational sources, including clusters and superclusters, in the region surrounding the central fluid, as in the actual universe, is a serious limitation of this model that is avoided by FRW cosmologies.  However, FRW cosmologies suffer from a flaw at the opposite extreme.  There are no local vacuums for FRW metrics that model a universe filled with matter.  Instead, in those models, space is filled with a continuum matter fluid that leaves no region of space empty.  This limits the utility of FRW cosmological models to analyze particle motion in the nearly empty space surrounding massive objects, just where the large scale homogeneity of the universe breaks down.\\

 \noindent The model considered here may thus be useful for the analysis of receding masses within a great void adjacent to a supercluster, relatively isolated from gravitational sources other than the supercluster. For example, for receding masses the line of best fit for data pairs of the form $(\rho^{*},\nu_{E}/\nu_{R})$ -- i.e., observed, or affine, distance versus ratios of emission frequency to reception frequency -- determines the slope, $\sqrt{\Lambda}/E_{0}$, and vertical intercept, $(1-\rho_{0}\sqrt{\Lambda})/E_{0}$ in  Eq.\eqref{hubble7}. \\ 
 
\noindent An estimate of the mass and radius of the supercluster determine $E_{0}$ and $\rho_{0}$, which, together with the slope and intercept of the preceding paragraph, lead to an estimate of the cosmological constant, $\Lambda$. Conversely, an estimate of $\Lambda$, together with observationally determined numerical values for the slope and intercept, determine the critical energy, $E_{0}$ and the critical radius $\rho_{0}$ where gravitational attraction and repulsion due to the cosmological constant exactly balance. The kinematic velocity, as a function of proper distance $\rho\gtrsim\rho_{0}$, is then determined by such measurements and Eq.\eqref{hubble3}. Note that the spectroscopic velocity is determined directly from observational data via Eq.\eqref{vspec}.\\

\noindent More generally, a numerical estimate for $\Lambda$, together with observationally determined numerical values for $E_{0}$ and $\rho_{0}$, may be used to calculate $M$ and $R$ through Eqs.\eqref{proper}, \eqref{r0}, and \eqref{rest} and numerical methods such as Newton's method for the determination of roots of a two-component function of the two variables $M$ and $R$.  In this way, the four different relative velocities are determined through direct calculation or through the relationships of Sec. 6.  We note also that the four-velocity of a radially receding mass is uniquely determined by its kinematic relative velocity via Eq.\eqref{vk1}.\\     

\noindent Schwarzschild-de Sitter space, with metric given by Eq.\eqref{metric1}, also serves as a model for the Local Group of galaxies in the far future.  As argued in \cite{krauss, newastronomy}, calculations show that the Local Group will remain gravitationally bound in the face of accelerated Hubble expansion, while more distant structures are driven outside of the cosmological horizon.  The Local Group, decoupled from the Hubble expansion, will be gravitationally bound and surrounded by a vacuum.\\

\noindent We note, in contrast to assertions in \cite{krauss}, that future cosmologists should in principle be able to detect the presence of a cosmological constant, provided they have the means to measure relative velocities of receding test particles, since the formulas calculated in the preceding sections for Fermi, kinematic, spectroscopic, and astrometric relative velocities all depend on $\Lambda$. The qualitative behavior of the relative velocities does not depend on special choices of the parameters. However, the cases of $\Lambda = 0$ and $\Lambda >0$ yield significantly different qualitative behaviors of the trajectories of outbound test particles.\\

\noindent \textbf{{\normalsize 9. Concluding Remarks}}\\

\noindent  Using global Fermi coordinates, we have calculated four geometrically defined velocities of radially receding test particles, relative to the central observer in Schwarzschild-de Sitter space:  Fermi, kinematic, spectroscopic, and astrometric relative velocities. The critical energy $E_{0}$, defined by Eq.\eqref{rest}, is a key parameter and plays multiple roles in this spacetime. It determines the redshift of a light signal received by the central observer and emitted from a static test particle at a point in space with radial coordinate $r_{0}$, where inward gravitational acceleration exactly balances the outward acceleration due to the cosmological constant (Sec. 7). In geometric units, $E_{0}$ is the maximum Fermi speed of a photon relative to the central observer (see Remark \ref{light}).  Receding test particles with energies in excess of $E_{0}$ may be regarded as having ``peculiar velocities'' while particles with energy $E_{0}$ obey a version of Hubble's law in the form of Eqs.\eqref{hubble3}, \eqref{hubble7}, and \eqref{hubble9}. The critical energy together with the cosmological constant, $\Lambda$, determines the redshift of light signals from receding masses, in general, as given by Eq.\eqref{hubble7}.\\

\noindent The Fermi relative velocity of a radially receding unit mass test particle, whose energy lies between energy $E_{0}$ and $\sqrt{2}\,E_{0}$, decreases under the influence of gravity near the central fluid, but far from the fluid (for $r>r_{0}$) the particle accelerates toward the cosmological horizon because of the influence of the cosmological constant.  Within this energy range, the qualitative behavior of the trajectory is consistent with  Newtonian mechanics.\\

\noindent  However, the trajectory of a receding test particle whose energy exceeds $\sqrt{2}\,E_{0}$ is more surprising.  The behavior of the Fermi relative velocity is essentially opposite to its Newtonian counterpart. As shown in Fig.\ref{fig1}, the high energy particle accelerates away from the central mass in the region dominated by gravity, surpassing the speed of light (by Proposition \ref{superluminal}), and then decelerates in the region dominated by the cosmological constant (where relative Fermi velocity of the low energy particle increases).  The effects of the gravitational field and the cosmological constant are reversed in this situation.\\

\noindent A similar, though not entirely analogous, phenomenon occurs in FRW matter (i.e. dust) dominated cosmologies. It was shown in \cite{gron} that, in an expanding universe, a test particle initially at rest relative to a distant observer accelerates toward the observer, according to that observer's proper time and distance measurements, i.e., in Fermi coordinates.\\

\noindent Other comparisons with FRW cosmologies can be made. Outside of the Hubble sphere in FRW cosmologies, the Fermi velocities of receding test particles, relative to the observer at the center of the sphere, exceed the speed of light (cf. Eq. (22) of \cite{gron}). In the model of a galaxy supercluster with surrounding vacuum considered in Sec. 8, the proper distance from the central observer to the cosmological horizon is of the same order of magnitude as the Hubble radius.  In contrast to that model, the Fermi velocity decreases to zero because $\partial/\partial t$ becomes null at the cosmological horizon, but the spectroscopic and kinematic velocities increase asymptotically to the speed of light at that distance (as shown in Fig.\ref{fig2}).  The same phenomena occur for the Local Group of galaxies surrounded by the vacuum that results from the Hubble flow, far into the future. \\

\noindent Hubble's law and the existence of superluminal relative velocities in FRW spacetimes have been used to support the interpretation that, in an expanding universe, galaxy clusters and superclusters are not merely flying apart from each other, space itself is expanding, e.g., \cite{confusion, gron}. But if a Hubble's law or the existence of superluminal Fermi relative velocities characterizes the expansion of space, then we have shown that space expands in the models considered here.  That seems implausible. Proposition \ref{superluminal2} shows that even in the static Schwarzschild spacetime, for which $\Lambda=0$, superluminal relative Fermi velocities necessarily exist.  In that case, the local mass distribution, represented by $A(0)$, in the vicinity of the observer determines the maximum possible Fermi relative velocity of a receding test particle.  In the case that $\Lambda >0$, the maximum relative Fermi velocity of a receding particle is determined by the critical energy, $E_{0}$.\\

\end{document}